\documentclass[]{llncs}
\pdfoutput=1

\usepackage{amsmath}
\usepackage{amsfonts}
\usepackage{amssymb}
\usepackage{stmaryrd}

\usepackage{etoolbox}
\usepackage{multirow}

\usepackage[pdftex]{graphicx}

\usepackage{macros}
\usepackage{url}

\usepackage{verbatim}
\usepackage{pdfsync}
\usepackage{numprint}
\npthousandsep{\;}
\npdecimalsign{.}

\newtheorem{myprob}[corollary]{Problem}

\usepackage{chngcntr}
\usepackage{apptools}
\AtAppendix{\counterwithin{lemma}{section}}

\usepackage[font=scriptsize]{subcaption}
\usepackage{wrapfig}

\usepackage[linesnumbered,ruled,noend]{algorithm2e}

\title{Optimal Continuous Time Markov Decisions}

\author{Yuliya Butkova, Hassan Hatefi, Holger Hermanns, Jan Kr\v{c}\'{a}l}
\institute{Saarland University -- Computer Science, Saarbr\"{u}cken, Germany}

\begin{document}

\maketitle

\begin{abstract}
  In the context of Markov decision processes running in continuous
  time, one of the most intriguing challenges is the efficient
  approximation of finite horizon reachability objectives. A multitude
  of sophisticated model checking algorithms have been proposed for
  this.
However, no proper benchmarking has
been performed thus far.

This paper presents a novel and yet simple solution: an algorithm,
originally developed for a restricted subclass of models and a
subclass of schedulers, can be twisted so as to become competitive with
the more sophisticated algorithms in full generality. As the second
main contribution, we perform a comparative evaluation of the core
algorithmic concepts on an extensive set of benchmarks varying over
all key parameters: model size, amount of non-determinism, time
horizon, and precision.




\end{abstract}


\section{Introduction}
\label{sec:intro}

Over the last two decades, a formal approach to quantitative
performance and dependability evaluation of concurrent systems has
gained maturity. At its root are continuous-time Markov chains (CTMC)
for which efficient and quantifiably precise solution methods
exist~\cite{BaierHHK03}. A CTMC can be viewed as a labelled transition
system (LTS) whose transitions are delayed according to exponential
distributions. CTMCs are stochastic processes and thus do not support
non-determinism. Non-determinism, often present in classical
concurrency and automata theory models, is useful for modelling
uncertainty or for performing optimisation over multiple 
choices. The genuine extension of CTMCs with non-determinism are
continuous time Markov decision processes (CTMDPs). The non-determinism is controlled by an object called \emph{scheduler} (also policy or strategy).

Prominent applications of CTMDPs include power management and
scheduling~\cite{QiuQP01}, networked, distributed
systems~\cite{GhemawatGL03,HaverkortHK00}, epidemic and population
processes~\cite{Lefevre81}, economy~\cite{BrunoDF81} and others.
Moreover, CTMDPs are the core semantic
model~\cite{DBLP:conf/apn/EisentrautHK013} underlying formalisms such
as generalised stochastic Petri
nets~\cite{Marsan:1994:MGS:561155}, Markovian
stochastic activity networks~\cite{DBLP:conf/pnpm/MeyerMS85} and
interactive Markov chains~\cite{DBLP:conf/fmco/HermannsK09}.

When model checking a CTMDP~\cite{DBLP:conf/cav/BuchholzHHZ11}, one
asks whether the behaviour of the model for \emph{some} schedulers (if we control the non-determinism) or for \emph{all} schedulers (if it is out of control) 
satisfies given performance or dependability criteria.
A large
variety of them can be expressed using logics such as
CSL~\cite{AzizSSB96}.
At the centre of 
model-checking problems for such criteria
is time bounded reachability: \emph{What is the maximal/minimal
	probability to reach a given set of states within a given time
	bound}. 
      Having an efficient approach for this optimisation (maximisation
      or minimisation) is crucial for successful large-scale
      applications. 

      In order to not discriminate against real situations, one
      usually assumes that the scheduler can base its decisions
      on \emph{any} available information about the past.
Restricting the information however tends to imply cheaper approximative algorithms~\cite{DBLP:conf/fossacs/NeuhausserSK09,DBLP:journals/tcs/BaierHKH05}.
For CTMDPs, we can distinguish (general) \emph{timed} optimal scheduling and (restricted) \emph{untimed} optimal scheduling~\cite{DBLP:journals/tcs/BaierHKH05,DBLP:journals/iandc/BrazdilFKKK13}. In the latter case, the scheduler has no possibility, intuitively speaking, to look at a clock measuring time.
Another distinction within timed optimality discussed in the literature is \emph{early} optimal scheduling (where every decision is frozen in between state changes~\cite{DBLP:conf/qest/NeuhausserZ10,DBLP:conf/fsen/HatefiH13}) and \emph{late} optimal scheduling (where every decision can change as time passes while residing in a state~\cite{DBLP:journals/cor/BuchholzS11,fearnley_et_al:LIPIcs:2011:3354}).
%
%

%

A handful of sophisticated algorithms have been suggested 
for timed optimality (partly for early optimality, partly for late optimality)
signifying both the importance and the difficulty of this
problem~\cite{DBLP:conf/qest/NeuhausserZ10,DBLP:journals/cor/BuchholzS11,fearnley_et_al:LIPIcs:2011:3354}. This
paper presents a substantially different algorithm addressing this
very problem. The approach is readily applicable to both early and
late optimality. It harvest a very efficient algorithm for untimed
optimality~\cite{DBLP:journals/tcs/BaierHKH05} originally restricted
to a subclass of models.  
By a simple twist, we make it applicable for the general timed optimality for arbitrary models.
As a second contribution, we present an
exhaustive empirical comparison of this novel algorithm with all other
published algorithms for the (early or late) timed optimality problem.
We do so on an extensive 
collection of scalable industrial and academic CTMDP benchmarks (that
we also make available). Notably, all earlier evaluations did compare
at most two algorithms on at most one or two principal cases.  We
instead cross-compare $5$ algorithms on $7$ application cases,
yielding a total of about 2350 distinct configurations.  The results
demonstrate that our simple algorithm is highly efficient across the
entire spectrum of models, except for some of the experiments where
extreme precision is required. On the other hand, no algorithm is
consistently dominating any other algorithm across the experiments
performed.



\paragraph{Related work.} 
Timed optimal scheduling has been considered for many decades both theoretically~\cite{Miller68,DBLP:journals/tcs/RabeS13} and practically by introducing approximative algorithms. Formal error bounds needed for verification have been studied only recently~\cite{DBLP:conf/qest/NeuhausserZ10,DBLP:conf/fsen/HatefiH13,fearnley_et_al:LIPIcs:2011:3354,DBLP:journals/cor/BuchholzS11}. 
Fragmentary empirical evaluations of some of the published
algorithms have been
performed~\cite{DBLP:conf/cav/BuchholzHHZ11,fearnley_et_al:LIPIcs:2011:3354,DBLP:conf/fsen/HatefiH13}. In a nutshell, the published knowledge boils down to 
\cite{DBLP:conf/qest/NeuhausserZ10}
$\lessdot_{_{\text{\cite{DBLP:conf/fsen/HatefiH13}}}}$
\cite{DBLP:conf/fsen/HatefiH13}
and
\cite{DBLP:conf/qest/NeuhausserZ10}
$\lessdot_{_{\text{\cite{DBLP:conf/cav/BuchholzHHZ11}}}}$
\cite{DBLP:journals/cor/BuchholzS11}
$\lessdot_{_{\text{\cite{fearnley_et_al:LIPIcs:2011:3354}}}}$
\cite{fearnley_et_al:LIPIcs:2011:3354},
%
%
where  $a \lessdot_{_{[\cdot]}} b$ denotes ``$b$ is shown empirically faster than $a$ in~$[\cdot]$''.
A substantial cross-comparison of the newest three algorithms~\cite{DBLP:journals/cor/BuchholzS11,fearnley_et_al:LIPIcs:2011:3354,DBLP:conf/fsen/HatefiH13} is however lacking.

\paragraph{Contribution of the paper.} 
The 
paper $(i)$ develops a novel and simple approximation method for time
bounded CTMDP reachability, $(ii)$ presents the first ever set of benchmarks
for CTMDP model checking, and $(iii)$ performs
an empirical evaluation across benchmarks and algorithms. The evaluation suggests  
that the optimal timing of decisions
for time bounded reachability can be solved effectively by a
rather straightforward algorithm, unless extreme precision is needed.
%


\section{Preliminaries}
\label{sec:ctmcp}
\begin{definition}
  A \emph{continuous-time Mar\-kov decision process} (CTMDP)\label{sym:C2} is a tuple $\C=(S,\Act,\bfR)$ where $S$ is a finite
  set of \emph{states}, $\Act$\label{sym:Act} is a finite set of \emph{actions}, 
 and $\bfR\label{sym:bfR2}: S \times \Act \times S \to
  \mathbb{R}_{\geq 0}$ is a rate function.
\end{definition}

We call an action $\acta$ \emph{enabled} in $s$, also denoted by $a \in \Act(s)$, if $\bfR(s,\acta,s') > 0$ for some $s'\in S$.
We require that all sets $\Act(s)$ are non-empty.
A continuous-time Markov chain (CTMC) is a CTMDP where all $\Act(s)$ are
singleton sets. 

For a given state $s$ and action $\acta \in\Act(s)$, we denote by $\exit(s,\acta) = \sum_{s'} \bfR(s,\acta,s')$ the \emph{exit} rate of $a$ in $s$. Finally, we let $\bfP(s,\acta,s') := \bfR(s,\acta,s') / \exit(s,a)$.




The operational behaviour of a CTMDP is 
like in a CTMC.
Namely, when performing a given action $\acta_0$ in a state $s_0$, the CTMDP waits for a transition, i.e. 
 waits for a delay $t_0$ chosen randomly according to an exponential distribution with rate $\exit(s_0,\acta_0)$. 
 The 
transition leads to a state $s_1$ again chosen randomly according to the probability distribution $\bfP(s_0,\acta_0,\cdot)$. When performing an action $a_1$ there, it similarly waits for time $t_1$ and makes a transition into a state $s_2$ and so on, 
forming an infinite \emph{run} $s_0 t_0 s_1 t_1 \cdots$. 

The difference to a CTMC lies in the  need to choose actions to perform, done by a \emph{scheduler}. 
%
%
There are two classes of schedulers, \emph{early} and \emph{late}.  Whenever entering a state, an early scheduler needs to choose and commit to a next action, whereas late schedulers may change such choices at any time later while residing in the state. 
In this paper we restrict w.l.o.g.~\cite{DBLP:phd/de/Neuhausser2010} to deterministic schedulers but we allow the decision to depend on the whole \emph{history} $s_0 t_0 \cdots t_{n-1} s_n$ 
so far.

\begin{definition}
A (timed late) randomised \emph{scheduler} is a measurable\footnote{Measurable with respect to the standard $\sigma$-algebra on the set of finite histories~\cite{DBLP:phd/de/Neuhausser2010}.} function $\sigma$ that to any history $h = s_0 t_0 \cdots t_{n-1} s_n$ and time $t \geq 0$ spent in $s_n$ so far assigns a distribution over enabled actions $\Act(s_n)$.
We call $\sigma$ \emph{early} if
$\sigma(h,t) = \sigma(h,t')$ for all $h, t,t'$;
%
and 
\emph{deterministic} if 
$\sigma(h,t)$ assign $1$ to some action $a$ for all $h,t$.
%
%
\end{definition}

We denote the set of all (timed) late or early schedulers by $\GM_\late$ and $\GM_\early$, respectively. We use these subscripts $\x \in \{\late,\early\}$ throughout the paper to distinguish between the late and the early setting. 
Furthermore, a scheduler $\sigma$ is called \emph{untimed} if $\sigma(h,t) = \sigma(h',t')$ whenever $h$ and $h'$ contain the same sequence of states. By $\TA$ we denote the set of all untimed schedulers. Note that $\TA \subseteq \GM_\early \subseteq \GM_\late$.

Fixing a scheduler $\sigma$ and an initial state $s$ in a CTMDP $\C$, we obtain the unique probability measure $\Prb[\C,s]{\sigma}$ over the space of all runs by standard definitions~\cite{DBLP:phd/de/Neuhausser2010}, denoted also by $\Prb[s]{\sigma}$ when $\C$ is clear from context.

\begin{myprob}[Maximum Time-Bounded Reachability]\label{prob:approx}
Let $\C = (S,\Act,\bfR)$, $G \subseteq S$ be a set of goal states, $T \in\Rsetpo$ a time bound, and $\x \in \{\late,\early\}$. 
Approximate the \emph{values} $\val{}_{\C} \in [0,1]^S$, where each $\val{}_{\C}(s)$ maximises the probability 
$$\val{}_{\C}(s) := \sup_{\sigma \in \GM_\x} \Pr[s]{\sigma}{\reach{T}{G}}$$
of runs $\reach{T}{G} = \{s_0 t_0 \cdots \mid \exists i:  s_i \in G \land \sum_{j=0}^{i-1} t_j \leq T\}$ reaching $G$ before $T$.
\end{myprob}

%
Whenever $\C$ is clear from 
context, we write $\val{}$.
%
We call $\sigma \in \GM_\x$ \emph{$\epsilon$-optimal} if $\Pr[s]{\sigma}{\reach{T}{G}} \geq \val{}(s) - \eps$ for all $s\in S$, and \emph{optimal} if it is $0$-optimal.
%
%

By minor changes, all results of the paper also address the dual problem of \emph{minimum} time bounded reachability that we omit to simplify the presentation.

\begin{remark} \label{rem:transf}
  There exists a value preserving encoding of early scheduling into
  late scheduling in CTMDPs~\cite{DBLP:journals/acta/RabeS11}. It has
  exponential space complexity (due to the number of induced
  transitions). This exponentiality does arise in practice, e.g.  for the stochastic job scheduling problem considered
  later. Therefore we  treat the two algorithmic
  settings separately. Early scheduling is  natural for models derived
  from generalised stochastic Petri nets or interactive Markov chains.
\end{remark}

\subsubsection*{}

\vspace*{-10mm}

\section{\GU: Optimal Time-Bounded Reachability Revisited}\label{sec:our}

\begin{algorithm}[b!]
	\SetAlgoLined
	\DontPrintSemicolon
	
	\SetKwInOut{Input}{input}
	\SetKwInOut{Output}{output}
\SetKwInOut{Par}{params}
	
	\Input{CTMDP $\C = (S,\Act,\bfR)$, goal states $G\subseteq S$, horizon $T \in \Rsetp$, scheduler class $\x \in \{ \late, \early \}$, and approximation error $\eps>0$}
	\Par{truncation error ratio $\kappa \in (0,1)$}
	\Output{vector $\valVar{}$ such that $\lVert \valVar{} - \val{} \rVert_\infty \leq \eps$ and $\lambda$}
	\BlankLine
	
	$\lambda \leftarrow$ maximal exit rate $\lambdamin$ in $\C$ \;
	\BlankLine
	
	\Repeat{\textnormal{${\lVert \ovvalVar{} - \unvalVar{} \rVert_\infty} \leq \eps \cdot (1-\kappa)$}}{
		$\C^\x_\lambda \leftarrow$ $\x$-uniformisation of $\C$  to the rate $\lambda$\;
		$\unvalVar{} \leftarrow $ approximation of the lower bound $\unval{}$ for $\C^\x_\lambda$ up to error $ \eps \cdot \kappa$\; 		$\ovvalVar{} \leftarrow $ approximation of the upper bound $\ovval{}$ for $\C^\x_\lambda$ up to error $\eps \cdot \kappa$\; 		
		$\lambda \leftarrow 2 \cdot \lambda$\;
	}
	\Return $\unvalVar{}, \lambda$
	
	\caption{\GU}
	\label{alg:unif}
\end{algorithm}

In this section, we develop a novel and simple algorithm for Problem~\ref{prob:approx}. 
We fix $\C = (S,\Act,\bfR)$, $G \subseteq S$, $T \in\Rsetpo$, $\x \in \{\late,\early\}$ and an approximation error $\eps >0$.
Furthermore, let $\lambdamin := \max_{s,a} \exit(s,\acta)$ denote the maximal exit rate in $\C$.

In contrast to existing methods, our approach does not involve discretisation.
The algorithm instead builds upon \emph{uniformisation}~\cite{jensen1953markoff} and \emph{untimed} analysis~\cite{DBLP:journals/tcs/BaierHKH05,DBLP:journals/iandc/BrazdilFKKK13,DBLP:journals/tcs/RabeS13}. It is outlined in Algorithm~\ref{alg:unif}.
Technically, it is based on an iterative computation of tighter and tighter lower and upper bounds on the values until the required precision is met. In the first iteration, a \emph{uniformisation rate} $\lambda$ is set to $\lambdamin$, in every further iteration its value is doubled. In every iteration, we compute a lower bound $\unval{}$ and an upper bound $\ovval{}$ by two types of untimed analyses on the CTMDP $\C^\x_\lambda$ obtained by uniformising $\C$ to the rate $\lambda$.
In the remainder of this section, we explain the individual steps of Algorithm~\ref{alg:unif}, and prove correctness and termination.

Informally, the lower bound is based on maximum time bounded reachability with respect to the untimed scheduler subclass~\cite{DBLP:journals/tcs/BaierHKH05}. The upper bound, similarly to the one in~\cite{DBLP:journals/cor/BuchholzS11}, is based on \emph{prophetic} untimed schedulers that yield higher value than timed schedulers by knowing in advance how many steps will be taken within time $T$.
The intuition is that an untimed scheduler can approximately observe
the elapse of time by knowing the count of steps taken and the
expected delay per every step. In uniformised models, these delay
expectations are identical across all states (forming a Poisson
process) and therefore allow easy access to the expected total elapsed
time. By uniformising the model with higher and higher uniformisation
rates, this implicit knowledge of untimed schedulers increases. On the
other hand, the knowledge of prophetic untimed schedulers decreases;
both approaching the power of timed schedulers.


\subsection{Uniformisation to $\C_\lambda^\x$}

CTMDP $\C$ may have transitions with very different rates across
different states and actions. Here, we discuss how to 
perform \emph{uniformisation} for such a model. This is a conceptually well-known idea~\cite{jensen1953markoff}.
%
Applying it to $\C$ intuitively makes transitions occur with a
higher rate $\lambda \geq \lambdamin$, uniformly across all states and
actions.

To ensure that uniformisation does not change the schedulable
behaviour, we need distinct uniformisation procedures for the early
and the late setting.  Late uniformisation is straightforward, it
adds self-loops to states and actions where needed.
\begin{definition}[Late uniformisation]
	For $\lambda \geq \lambdamin$ we define the \emph{late uniformisation of $\C$ to rate $\lambda$} as a CTMDP $\C_\lambda^\late = (S,\Act,\bfR_\lambda^\late)$ where 
	$$\bfR_\lambda^\late(s,\acta,s') := \begin{cases}
	\bfR(s,\acta,s') & \text{if $s \neq s'$,}\\
	\lambda - 
	\sum \limits _{s''\neq s} \bfR(s,\acta,s'') 
& \text{if $s = s'$.}
	\end{cases}$$
\end{definition}

\begin{example} For the fragmentary CTMDP $\C$ depicted below on the left,
   its late uniformisation to rate $4.5$ is depicted in the middle.
%
%
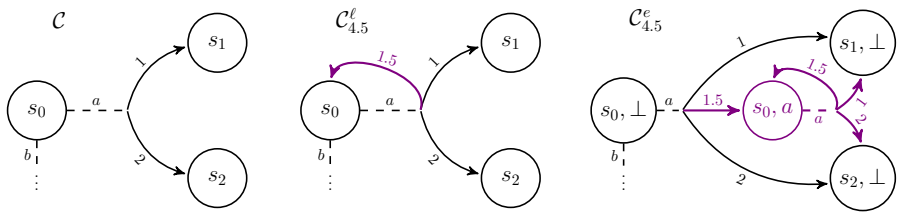
\begin{center}
	\begin{tikzpicture}[->,>=stealth',shorten >=1pt,auto,node distance=2cm,semithick, scale=0.6, transform shape]	  
	
	\tikzstyle{state}=[circle, fill=white,text=black, draw=black, minimum size=1.3cm, font=\Large, transform shape]
	\tikzstyle{markovian}=[draw=none,fill=none, inner sep=0, outer sep=0]

  \node[state]	(s0) {$s_0,\bot$};
  \node[font=\Large] at($(s0)+(0.5,2cm)$) {$\C^\early_{4.5}$};
  \node[markovian] 	(s0_a)	at($(s0)+(1.3cm,0)$)	{};
  \node[markovian] 	(s0_b)	at($(s0)-(0,1.3cm)$)	{};
  \node[markovian] 	(dots)	at($(s0_b)-(0,0.2cm)$)	{$\vdots$};
  \node[state]		(s1)	at($(s0_a)+(4cm,1.5cm)$) {$s_1,\bot$};
  \node[state]		(s2)	at($(s0_a)+(4cm,-1.5cm)$) {$s_2,\bot$};
  \node[state, draw=violet, text=violet]		(s0a_cp)	at($(s0_a)+(2cm, 0)$) {$s_0,a$};
  \node[markovian] 	(s0a_cpa)	at($(s0a_cp)+(1.4cm,0)$)	{};

		  \path (s0)
		  			edge [-,midway, above, dashed] node {$a$} 		(s0_a)
		  			edge [-,midway, left, dashed] node {$b$} 		(s0_b)
  		(s0_a)		edge [midway, above, sloped, bend left=30] node {1} 	(s1)
  					edge [midway, below, sloped, bend right=30] node {2} 	(s2)
  					edge [midway, above, violet, thick] node {1.5}	(s0a_cp)
		(s0a_cp)	edge [-,midway, below, dashed, violet, thick] node {$a$}			(s0a_cpa)
  		(s0a_cpa)	edge [midway, below, sloped, violet, thick, bend right=20] node {1}	(s1)
  					edge [midway, above, sloped, violet, thick, bend left=20] node {2} 	(s2)
  					edge [midway, above, sloped, bend right=90, violet, thick] node {1.5}	(s0a_cp)
  ;

  \node[state]	(s20) at($(s0)+(-13cm,0)$) {$s_0$};
  \node[font=\Large] at($(s20)+(0.5,2cm)$) {$\C$};
  \node[markovian] 	(s20_a)	at($(s20)+(2cm,0)$)	{};
  \node[markovian] 	(s20_b)	at($(s20)-(0,1.3cm)$)	{};
  \node[markovian] 	(dots2)	at($(s20_b)-(0,0.2cm)$)	{$\vdots$};
  \node[state]		(s21)	at($(s20_a)+(2cm,1.5cm)$) {$s_1$};
  \node[state]		(s22)	at($(s20_a)+(2cm,-1.5cm)$) {$s_2$};

  \path (s20)		
  				edge [-,midway, above, dashed, sloped] node {$a$} 		(s20_a)
				edge [-,midway, left, dashed] node {$b$} 		(s20_b)
  		(s20_a)	edge [midway, above, sloped, bend left=30] node {$1$}	(s21)
  				edge [midway, below, sloped, bend right=30] node {$2$} 	(s22)
  ;

  \node[state]	(s10) at($(s0)+(-6.5cm,0)$) {$s_0$};
  \node[font=\Large] at($(s10)+(0.5,2cm)$) {$\C^\late_{4.5}$};
  \node[markovian] 	(s10_a)	at($(s10)+(2cm,0)$)	{};
  \node[markovian] 	(s10_b)	at($(s10)-(0,1.3cm)$)	{};
  \node[markovian] 	(dots1)	at($(s10_b)-(0,0.2cm)$)	{$\vdots$};
  \node[state]		(s11)	at($(s10_a)+(2cm,1.5cm)$) {$s_1$};
  \node[state]		(s12)	at($(s10_a)+(2cm,-1.5cm)$) {$s_2$};

  \path (s10)		edge [-,midway, above, dashed, sloped] node {$a$} 		(s10_a)
		  			edge [-,midway, left, dashed] node {$b$} 		(s10_b)
  		(s10_a)		edge [midway, above, sloped, bend left=30] node {$1$} (s11)
  					edge [midway, below, sloped, bend right=30] node {$2$} (s12)
  					edge [midway, above, sloped, bend right=90, violet, thick] node {1.5}	(s10)
   ;
\end{tikzpicture}
\end{center}
 Using the same transformation for
  the early setting would give the scheduler the spurious possibility
  to ``reconsider'' the choice of the action in a state whenever a
  newly added self-loop is taken. To exclude that possibility, early
  uniformisation introduces a copy state $(s,a)$ for each state $s$
  and action $a$ so as to ``freeze'' the commitment of choosing action
  $a$ until the next state change occurs.  The construction is shown
  on the right. States of the form $(s,\bot)$ correspond to the
  original states, i.e.~those where no action has been committed to
  yet.
\end{example}

\begin{definition}[Early uniformisation]
For $\lambda \geq \lambdamin$, the \emph{early uniformisation of $\C$ to rate $\lambda$} is a CTMDP $\C_\lambda^\early = (S \times (\{\bot\} \cup \Act) ,\Act,\bfR_\lambda^\early)$ where for every state $(s,\cdot)$, action $a\in \Act$, and every successor state $(s',\circ)$ we have
	$$\bfR_\lambda^\early((s,\cdot),\acta,(s',\circ)) := \begin{cases}
	\bfR(s,\acta,s') & \text{if $\circ=\bot$,}\\
	\lambda - \exit(s,\acta)
 & \text{if $\circ=a,s=s'$,}\\
	0 & \text{elsewhere}.
	\end{cases}$$
\end{definition}
Uniformisation preserves the value of time-bounded reachability for both early~\cite{DBLP:phd/de/Neuhausser2010} and late schedulers~\cite{Miller68}. 
\begin{lemma}\label{lemma:unif-preserves}
$\forall \lambda \geq \lambdamin. \ \val{}_\C = \val{}_{\C_\lambda^\x}$, i.e.~uniformisation preserves the value.
\end{lemma} 
As a result, we can proceed by  bounding the values of $\C_\lambda^\x$ for large enough $\lambda$ instead of bounding the values of the original CTMDP $\C$.


\subsection{Lower and upper bounds on the value of $\C_\lambda^\x$}

We now fix a $\lambda$ and consider a uniform CTMDP $\C_\lambda^\x$. We denote by $\reach[=i]{T}{G}$ the subset of runs $\reach{T}{G}$ reaching the target where exactly $i$ steps are taken up to time $T$. 
With this, we define the bounds by ranging over $\TA$ schedulers in $\C_\lambda^\x$:
\begin{align*}
\unval{s} := \sup_{\sigma \in \TA}
\sum_{i=0}^\infty  \Pr[s]{\sigma}{\reach[=i]{T}{G}},
\qquad
\ovval{s} := \sum_{i=0}^\infty 
\sup_{\sigma \in \TA} \Pr[s]{\sigma}{\reach[=i]{T}{G}}.
\end{align*}
Since all $\reach[=i]{T}{G}$ are disjoint and $\reach{T}{G} = \bigcup_{i\in \Nseto} \reach[=i]{T}{G}$, the value $\unval{}$ is the optimal reachability probability of
standard untimed schedulers on the uniformised model. It will serve
as a lower bound on the values $\val{}$.
The value $\ovval{}$, on the other hand, which has the supremum and summation swapped, does not correspond to the value of any realistic scheduler. Intuitively, it is the value of a \emph{prophetic} untimed scheduler, which for each particular run knows how many steps will be taken (as for every $i$, a different standard scheduler $\sigma$ may be used). This knowledge makes the scheduler more powerful than any other timed one:

\begin{lemma}\label{lem:lower-upper}
	It holds that $\vale{}\leq\vall{}$, and for any CTMDP $\C^\x_\lambda$, $\unval{}\leq\val{}\leq\ovval{}$.
\end{lemma}

\subsubsection{Approximating the bounds.}
Since $\unval{}$ and $\ovval{}$ are defined via infinite summations,
we need to approximate these bounds. We do so by  iterative
algorithms truncating the sums. This is what is computed in line 4 and 5 of Algorithm~\ref{alg:unif}. Each  truncation induces an error of up
to $\eps \cdot \kappa$.

Let $\pois(k)$ denote the Poisson distribution with parameter $\lambda T$ at point $k$, i.e.~the probability that exactly $k$ transitions are taken in the CTMDP $\C_\lambda^\x$ before time $T$. Furthermore, let $N = \lceil \lambda T \mathrm{e}^2 - \ln(\eps \cdot \kappa) \rceil$, where $\mathrm{e}$ is the Euler's number.
We recursively define for every $0\leq k \leq N$ and every state $s$, functions
\begin{align*}
\unvalVar[k]{s} &= 
\begin{cases}
0 & \text{if $k = N$,} \\
\sum_{i=k}^{N-1} \pois(i) & \text{if $k < N$ and $s \in G$,} \\
\max_\acta \sum_{s'} \bfP^\x_\lambda(s,\acta,s') \cdot \unvalVar[k+1]{s'} & \text{if $k < N$ and $s \not\in G$,}
\end{cases} \end{align*}
\begin{align*}
\ovvalprimeVar[k]{s} &= 
\begin{cases}
0 & \text{if $k = N$,} \\
1 & \text{if $k < N$ and $s \in G$,} \\
\max_\acta \sum_{s'} \bfP^\x_\lambda(s,\acta,s') \cdot \ovvalprimeVar[k+1]{s'} \;\; & \text{if $k < N$ and $s \not\in G$,} 
\end{cases} \\
\ovvalVar[k]{s} & = \sum_{i=k}^{N-1} \pois(i) \cdot \ovvalprimeVar[(N-1)-(i-k)]{s}\text{,} 
\end{align*}
where $\bfP^\x_\lambda$ denotes the transition probability matrix of $\C^\x_\lambda$.

\begin{lemma}\label{lem:approximations}
In any CTMDP $\C^\x_\lambda$, $\lVert \unvalVar[0]{} - \unval{} \rVert_\infty \leq \eps \cdot \kappa$ and $\lVert \ovvalVar[0]{} - \ovval{} \rVert_\infty \leq \eps \cdot \kappa$.
\end{lemma}

We compute $\unvalVar[0]{}$ as in the untimed
analysis of uniform models~\cite{DBLP:journals/tcs/BaierHKH05}, which in turn agrees with the
standard ``uniformisation'' algorithm for CTMCs when 
the maximisation is dropped. The computation of $\ovvalprimeVar[k]{}$ is analogous to step-bounded reachability for \emph{discrete-time} Markov decision processes, where the reachability probabilities for different step-bounds are weighted by the Poisson distribution in the end in $\ovvalVar[0]{}$. Both vectors can be 
computed in time $O(N\cdot|S|^2\cdot|Act|)$.

\paragraph{Numerical Aspects.}
In practice also $\unvalVar[0]{}$ and $\ovvalVar[0]{}$ can only be approximated due to presence of $\pois(k)$.
For details how the overall error bound is met in an analogous setting, see~\cite{DBLP:journals/iandc/BrazdilFKKK13}. 
For high values of $\lambda$ and thus also $N$, the Poisson
values $\pois(k)$ are low for most $0\leq k < N$ and also
the values in $\bfP^\x_\lambda$ get close to $1$ when on the diagonal and to $0$ when off-diagonal. Where high precision is required and thus high $\lambda$ may be needed, attention has
to be paid to numerical stability.
%



\subsection{Convergence of the bounds for increasing $\lambda$}

An essential part for the correctness of Algorithm~\ref{alg:unif} is
its convergence:
\begin{lemma}\label{lem:converge}
	We have $\displaystyle \lim_{\lambda \to \infty} g_\lambda \to 0$ where 
	$g_\lambda$ denotes the gap $\lVert \unval{} - \ovval{} \rVert_\infty$ in $\C^\x_\lambda$.
\end{lemma}

\paragraph{Proof Idea.} We here provide an intuition of the core of the proof, namely why uni-\linebreak
\begin{wrapfigure}[10]{r}{0.35\textwidth}
	\centering
	\vspace*{-11mm}
	\begin{tikzpicture}[
	params/.style={right,font=\footnotesize},
	mark=none, smooth
]
	
	
	\begin{axis}[
	ymin=0,ymax=1,xmin=2, xmax=18,
	axis y line=left,axis x line=bottom,
	y axis line style={-},
	width=5cm,
	legend style={font=\scriptsize,rounded corners,at={(0.5,1)}},
	grid style = {
		dash pattern = on 0.05mm off 1mm,
		line cap = round,
		black,
		line width = 0.5pt
	}
	]
	
	\addplot[black,line width=0.75pt,dashed] coordinates { (1.0,1.72115629955841E-4) (2.0,0.0036598468273436845) (3.0,0.018575936222140765) (4.0,0.05265301734371111) (5.0,0.10882198108584873) (6.0,0.18473675547622795) (7.0,0.27455504669039554) (8.0,0.3711630648201266) (9.0,0.46789642362528455) (10.0,0.5595067149347875) (11.0,0.6424819975720744) (12.0,0.7149434996833688) (13.0,0.7763281831885006) (14.0,0.8270083921179286) (15.0,0.8679381437122795) (16.0,0.9003675995129539) (17.0,0.9256360201854197) (18.0,0.9450363585048951) (19.0,0.9597373176593901)}
	 node [pos = 0.35, above, font=\scriptsize,rotate=50]{$\lambda=0.5$};
	
	\addplot[black,line width=1pt] coordinates{ (1.0,-1.2114427351103024E-16) (2.0,2.612879177940703E-17) (3.0,3.8442426545715423E-16) (4.0,1.1364130133830198E-15) (5.0,3.2000641335515783E-10) (6.0,1.4815276330662235E-6) (7.0,4.303725949803462E-4) (8.0,0.017108313035132633) (9.0,0.15822098918642885) (10.0,0.5132987982791489) (11.0,0.8417213299399128) (12.0,0.9721362601094794) (13.0,0.9972495916326934) (14.0,0.9998389428252875) (15.0,0.9999940754596642) (16.0,0.999999855797755) (17.0,0.9999999975699225) (18.0,0.9999999999705174) (19.0,0.9999999999997337)}
		 node [pos = 0.45, right,font=\scriptsize]{$\lambda=10$};
	
%
	
	
	\end{axis}
	
	\node[font=\scriptsize] at (3.2,0.2) {time};
	
	\end{tikzpicture}
\end{wrapfigure}

\vspace{-5.7ex}\noindent formisation with higher $\lambda$ increases the power of untimed schedulers and decreases the power of prophetic ones: The count of transitions
 taken so far gives untimed schedulers approximate knowledge of how much time has elapsed. In situations with the \emph{same expectation} of elapsed time, a higher uniformisation rate induces a \emph{lower variance} of elapsed time. On the right, we illustrate comparable situations for different uniformisation rates, after $5$ transitions with rate $0.5$ and after $100$ transitions with rate $10$. Both depicted cumulative distribution functions of elapsed time have expectation $10$ but the latter is way steeper, providing a more precise knowledge of time.

At the same time prophetic schedulers on the high-rate uniformised model are less powerful than on the original one. When taking decisions, the future evolution is influenced by two types of randomness: (a) continuous timing, i.e.~how many further transitions will be taken before the time horizon and (b) discrete branching, i.e.~which transitions will be taken.
Even though the value stays the same for arbitrary $\lambda$, the ``source of'' randomness for high $\lambda$ shifts from (a) to (b). Namely, 
the distribution of the number of future transitions also becomes steeper for higher $\lambda$, thus being ``less random'' by having smaller coefficient of variation.
At the same time, 
the discrete branching for higher $\lambda$ influences more the number of \emph{actual} transitions taken (i.e.~transitions that are not the added self-loops).
As a result, the advantage of the prophetic scheduler is only little  as (i) it boils down to observing the outcome of a less and less random choice and (ii) the observed quantity has little impact on how many actual transitions are taken.

$ $

As a result of Lemma~\ref{lem:converge}, we obtain that Algorithm~\ref{alg:unif} terminates. Its correctness follows from Lemma~\ref{lemma:unif-preserves},~\ref{lem:lower-upper} and~\ref{lem:approximations}, all summarized by the following theorem.
 
\begin{theorem}
 	Algorithm~\ref{alg:unif} computes an approximation of $\val{}$ up to error $\eps$.
\end{theorem}

\begin{remark}
	Algorithm~\ref{alg:unif} determines a sufficiently large $\lambda$ in an exponential search fashion. In practice, this approach is  efficient w.r.t.~the total number $I$ of \emph{iterations} needed, i.e.~the total number of times $\unvalVar[k]{}$ and $\ovvalprimeVar[k]{}$ are computed from $\unvalVar[k+1]{}$ and $\ovvalprimeVar[k+1]{}$. Namely, in practice the error monotonously decreases when the rate increases (not in theory but we never encountered the opposite case on our extensive experiments.) As a result, $\lambda$ found by
	Algorithm~\ref{alg:unif} satisfies $\lambda <  2\cdot
	\lambda^\ast$ where $\lambda^\ast$ is the minimal sufficiently large rate.
%
As the number of iterations needed for one approximation is linear in the uniformisation rate used, we have $I = 2 I_\lambda < 4 \cdot I_{\lambda^\ast}$, where each $I_{\lambda'}$ denotes the number of iterations needed for the computation for the fixed rate $\lambda'$.
\end{remark}

\subsection{Extracting the scheduler}
By computing the lower bound, Algorithm~\ref{alg:unif} also produces~\cite{DBLP:journals/tcs/BaierHKH05} an untimed scheduler $\sigma^\x_\lambda$ that is $\varepsilon$-optimal on the \emph{uniformised} model $\C^\x_\lambda$.
In the \emph{original} CTMDP $\C$, we cannot use $\sigma^\x_\lambda$ directly as its choices are tailored to the high rate $\lambda$.
We can however use a \emph{stochastic update} scheduler attaining the same value.
%
Informally, a (timed) \emph{stochastic update scheduler} $\sigma = (\mathcal{M}, \sigma_u, \pi_0)$ operates over a countable set $\mathcal{M}$ of memory elements where the initial memory value is chosen randomly according to the distribution $\pi_0$ over $\mathcal{M}$. The stochastic update function $\sigma_u$, given the current memory element, state, and the time spent there, defines a distribution specifying the action to take and how to update the memory.
Intuitively, the stochastic update is used for simulating the high-rate transitions that would be taken in $\C_\lambda^{\x}$; their total count so far is stored in the memory. 
For a formal definition of stochastic update and the construction, see the Appendix.



\begin{lemma}\label{lemma:stoch-update} The values $(\unvalVar[k]{})_{0\leq k \leq N}$ computed by Algorithm~\ref{alg:unif} for given $\C$, $\x$, and $\eps > 0$ yield a stochastic update scheduler $\widetilde{\sigma}^{\x}_{\mathit{TD}}$ that is $\varepsilon$-optimal in $\C$.
\end{lemma}

\vspace*{-5mm}

\section{Existing Algorithms}

This  section briefly reviews  the various published algorithms solving Problem~\ref{prob:approx}. In contrast to Algorithm~\ref{alg:unif} (called \GU or \GUS for short), they all discretise time into a finite number of time points $t_0, t_1, \ldots, t_n$ where $t_0 = 0$ and $t_n = T$.
They iteratively approximate the values $\val[t_i]{s} := \sup_{\sigma \in \GM_\x} \Pr[s]{\sigma}{\reach{t_i}{G}}$ when $t_i$ time units remain at state $s$. 
%
%
%
Three  different iteration concepts have been proposed,
each approximating
$\val[t_{i+1}]{s}$
from approximations of 
$\val[t_i]{s'}$.

\paragraph{Exponential approximation -- early~\cite{DBLP:conf/qest/NeuhausserZ10,DBLP:conf/fsen/HatefiH13}.} Assuming  equidistant points $t_i$ one can approximate the (early) value function by piece-wise exponential functions. A $k$-order approximation  considers only runs where at most $k$ steps are taken between any two time points. This can yield an a priori error bound. 
%
%
%
The higher $k$, the less time points are required for a given precision, but the more computation is needed per time point. We refer to these algorithms by $\ES$-$k$ or $\ESS$-$k$ for short. Only $\ESS$-$1$~\cite{DBLP:conf/qest/NeuhausserZ10} and $\ESS$-$2$~\cite{DBLP:conf/fsen/HatefiH13} have been implemented so far.
%

\paragraph{Polynomial approximation -- late~\cite{fearnley_et_al:LIPIcs:2011:3354}.} Another way to approximate the (late) value function on equidistant time points uses polynomials. As before, the higher the degree of the polynomials, the higher is the computational effort, but the number of 
discretised time points 
required to assure an a priori error bound decreases. 
We call these algorithms $\PS$-$k$ or $\PSS$-$k$ in the sequel, only
$\PSS$-1, $\PSS$-2, and $\PSS$-3 have been implemented. Among these,
$\PSS$-2 has better worst-case behaviour, but
$\PSS$-3 has been reported to often perform better in practice.

\paragraph{Adaptive discretisation -- late~\cite{DBLP:journals/cor/BuchholzS11}.} 
This approach is not based on an a priori error bound but instead computes both under- and over-approximations of the values $\val[t_i]{s}$. This allows one to lay out the time points \emph{adaptively}. 
Depending on the shape of the value function, the time step can be prolonged until the error allowed for this step is reached. This greatly reduces the number of time points, relative to the worst case. We refer to this algorithm as \AS or \ASS.

\vspace*{-5mm}

\section{Empirical Evaluation and Comparison}
\label{sec:eva}

\definecolor{green_c}{RGB}{34,139,34}

\pgfplotscreateplotcyclelist{earlylist}{
  violet, very thick, dotted, mark=*, mark size=1\\%
  violet, very thick, dashed, mark=*, mark size=1\\%
  green_c, very thick, solid, mark=*, mark size=1\\%
    purple, thick, dotted, domain=\pgfkeysvalueof{/pgfplots/xmin}:\pgfkeysvalueof{/pgfplots/xmax}\\%
}

\pgfplotscreateplotcyclelist{latelist}{
	violet, line width=0.5pt, very thick, dotted, mark=*, mark size=1\\%
	violet, line width=0.5pt, very thick, dashed, mark=*, mark size=1\\%
	teal, line width=0.5pt, very thick, dotted, mark=*, mark size=1\\%
	teal, line width=0.5pt, very thick, dashed, mark=*, mark size=1\\%
	green_c, line width=0.5pt, very thick, solid, mark=*, mark size=1\\%
    purple, thick, dotted, domain=\pgfkeysvalueof{/pgfplots/xmin}:\pgfkeysvalueof{/pgfplots/xmax}\\%
}

\pgfplotsset{unif/.style={green_c, line width=0.5pt, very thick, solid, mark=*, mark size=1}}
\pgfplotsset{as/.style={teal, line width=0.5pt, very thick, dotted, mark=*, mark size=1}}
\newcommand*{\plotscale}{0.85}
\newcommand*{\plotratio}{0.32}
\newcommand*{\plotscalestate}{0.85}
\newcommand*{\plotratiostate}{0.32}

In this section we present an exhaustive empirical comparison of
the different algorithmic approaches discussed. 

\vspace*{-3mm}
\subsubsection{Benchmarks.} 
The experiments are
performed on a diverse collection of 
published benchmark models. This collection is the first of its kind for CTMDP, as far as we know
and contains the following parametrised models:
\begin{description}
\item [PS-$K$-$J$] The \emph{Polling System} case~\cite{GuckHHKT13,TimmerPS13} consists of two stations and
  one server. Incoming requests of $J$ types are
  buffered in two queues of size $K$ each, until they are processed by the server and delivered  to
  their station. We consider the undesirable states with both queues being full to form the goal state set.
  

\item [QS-$K$-$J$] The \emph{Queuing System}~\cite{HatefiH12} stores requests
  of $J$ different types into two queues of size $K$. Each queue is
  attached to a server. Two servers fetch requests from their
  corresponding queues and process them. One of them can
  non-deterministically decide to insert a request after processing
  into the other server's queue. Goal states are again those with both
  queues full.
  

\item[DPMS-$K$-$J$] The \emph{Dynamic Power Management System}~\cite{QiuQP01} is 
a CTMDP model of the internals of a Fujitsu disk drive.
The model consists
of four 
components: service requester (SR), service
queue (SQ), service provider (SP), and power manager (PM). SR generates tasks of $J$ types differing in energy demand that are buffered by the queue SQ of size $K$. Afterwards they are delivered to SP to be processed. SP can work in different modes ranging from sleep and stand-by to full processing mode, 
selected by PM. 
We define a state as goal if the queue of at least one task type is full.


\item[GFS-$N$] The \emph{Google File System}~\cite{GhemawatGL03,Guck12}
  splits files into chunks of equal size, each chunk is maintained by
  one of $N$ chunk servers. 
%
%
  We fix the number of chunks a server may store to \numprint{5000}
  and the total number of chunks to \numprint{100000}. 
%
  While other benchmarks start in optimal conditions, the GFS starts in the broken state where no chunk is stored. 
A state is defined as goal if the system is back up and for each chunk at least one copy is available.
  

\item[FTWC-$N$] The \emph{Fault Tolerant Workstation Cluster}~\cite{HaverkortHK00}, originally described by a
GSPN, models two networks of $N$ workstations each, interconnected  by a switch. The two switches communicate via a backbone. Workstations, switches, and the backbone fail after exponentially distributed delays, and can be repaired only one at a time. 
We define a state as goal if in total less than $N$ workstations are operational and connected to each other.
  

\item[SJS-$M$-$J$] The \emph{stochastic job scheduling}~\cite{BrunoDF81} models a  multiprocessor architecture running a sequence of
independent jobs. It consists of $M$ identical processors and
$J$ jobs, where each job's service time is governed by an
exponential distribution. As goal we define the desirable states with all jobs completed. 

\item[ES-$K$-$R$] The \emph{Erlang Stages} is a synthetic model with known characteristics~\cite{DBLP:conf/tacas/ZhangN10}. It has two different paths to reach the goal state: a fast but 
risky path or a slow but sure path. The slow path is an Erlang chain of length $K$ and rate~$R$.

  

\end{description}

\vspace*{-6mm}
\subsubsection{Implementation aspects.}
Unbiased performance evaluation of algorithms  originally developed by different researchers is not easy even with all original implementations at hand. Namely, they may use different programming languages or rely on different platforms with incomparable performance and memory management. 
However, reimplementing a published algorithm may induce unfairness as the original implementation may use specific data structures or other 
optimisations that go beyond what is explained in the respective publication.

\begin{sloppypar}
We adapted/implemented all algorithms in C/C++, trying to avoid the shortcomings. We used a common infrastructure 
from the IMCA/MAMA toolset~\cite{GuckHHKT13}. Thus, we could directly use the original IMCA implementations of \ES-1 and of \ES-2~\cite{DBLP:conf/fsen/HatefiH13}. 
The original implementation~\cite{DBLP:conf/cav/BuchholzHHZ11} of \AS in MRMC~\cite{KatoenZHHJ11} needed only minor adaptations, as MRMC uses a data structure identical to ours. 
Finally, for \PS, we closely followed the original Java code~\cite{fearnley_et_al:LIPIcs:2011:3354}. Our C version clearly outperforms the original Java version. 
\end{sloppypar}

                
We implemented all algorithms with standard double precision arithmetic, observing no issues with numerical stability in our experiments. All values computed by different algorithms lie within the expected precision from each other.

We used parameter values $k^{\max}=10$ and $\omega = 0.1$ for \ASS, as recommended. We always ran both adaptive and non-adaptive variant of \ASS and display 
the better results (mostly adaptive).
Based on our tests, we fixed $\kappa := 0.1$ for \GUS. 

\vspace*{-2mm}

\begin{figure}[t]
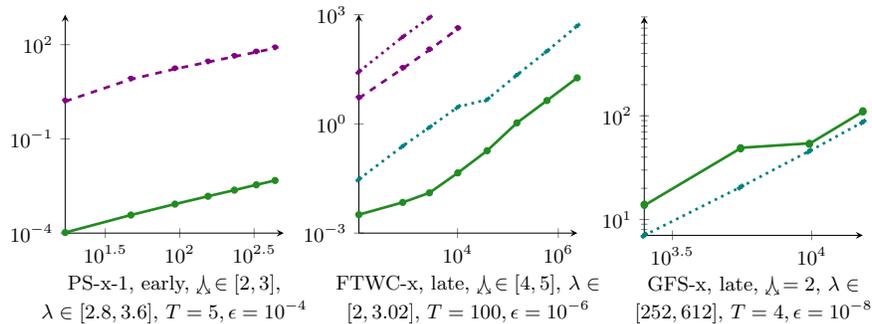

	\centering
	\begin{subfigure}[t]{\plotratiostate\textwidth}
		\centering
		\begin{tikzpicture}[scale=\plotscalestate]
		\input{figures/PollingSystem_states_early_jt1_t_5_e-04.tex}
		\end{tikzpicture}
		\label{fig:space:a}
	\end{subfigure}%
	\begin{subfigure}[t]{\plotratiostate\textwidth}
		\centering
		\begin{tikzpicture}[scale=\plotscalestate]
		\input{figures/FTWC_Late_States_t=100_e=1e-06.tex}
		\end{tikzpicture}
		\label{fig:space:c}
	\end{subfigure}%
	\begin{subfigure}[t]{\plotratiostate\textwidth}
		\centering
		\begin{tikzpicture}[scale=\plotscalestate]
		\input{figures/GFS_Late_States_t=4_e=1e-08.tex}
		\end{tikzpicture}
		\label{fig:space:d}
	\end{subfigure}%
	\vspace*{-5mm}
	\caption{Selected experiments: Increasing state space size.}
	\label{fig:space}
\vspace*{-12mm}
\end{figure}
\subsubsection{Empirical Results.}

In this section we present 
our empirical observations.
We consider early and late scheduling problems separately (because the encoding mentioned in  Remark~\ref{rem:transf} of Section~\ref{sec:ctmcp}, is exponential); only \GU can be directly run on both problems.
All experiments were run on a single core of Intel Core i7-4790 with 16GB of RAM, 
computing a total of about 2350 data points.
%


The memory requirements of all the considered algorithms do not deviate considerably and thus are not reported. This echoes that all space complexities are linear in the model size. 
%
%
We encountered no significant impact of additional 
dependencies of \PS
on a hidden model parameter 
(number of ``switching points'', coarsely bounded in~\cite{fearnley_et_al:LIPIcs:2011:3354}). 

In the following, we focus on the time requirements.
We first show plots of a few selected experiments that represent well our general observations. Later, we give a short summary of all experiments.
All plots presented below use logarithmic scale for the runtime
(in seconds). Some data points are missing as we applied a time limit of 15 minutes for every computation and also because the original implementation of \ES-2 cannot handle models with more than two actions per state.
We use symbol $\choices$ to denote the maximal number of action choices and $\lambda$ for the maximal exit rate. We use the symbol ``x'' whenever the varying
parameter is a part of the model name, e.g. PS-2-x.


\begin{figure}[t]
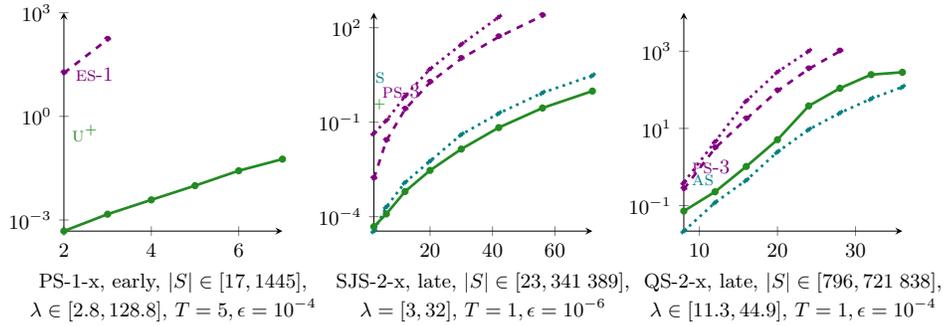

  \centering
  \begin{subfigure}[t]{\plotratio\textwidth} 
    \centering
    \begin{tikzpicture}[scale=\plotscale]
      \input{figures/PollingSystem_fanout_early_qs1_t_5_e-04.tex} 
    \end{tikzpicture}
    \label{fig:choice:a}
  \end{subfigure}%
  ~ 
  \begin{subfigure}[t]{\plotratio\textwidth} 
    \centering
    \begin{tikzpicture}[scale=\plotscale]
      \input{figures/SJS_Late_Fanout_t=1_e=1e-06.tex} 
    \end{tikzpicture}
    \label{fig:choice:b}
  \end{subfigure}%
  ~
  \begin{subfigure}[t]{\plotratio\textwidth} 
    \centering
    \begin{tikzpicture}[scale=\plotscale]
      \input{figures/QueuingSystem_Late_Fanout_t=1_e=1e-04.tex}
    \end{tikzpicture}
    \label{fig:choice:c}
  \end{subfigure}%
\vspace*{-5mm}
  \caption{Selected experiments: Increasing number of action choices.}
  \label{fig:choice}
\end{figure}

\begin{figure}[b]
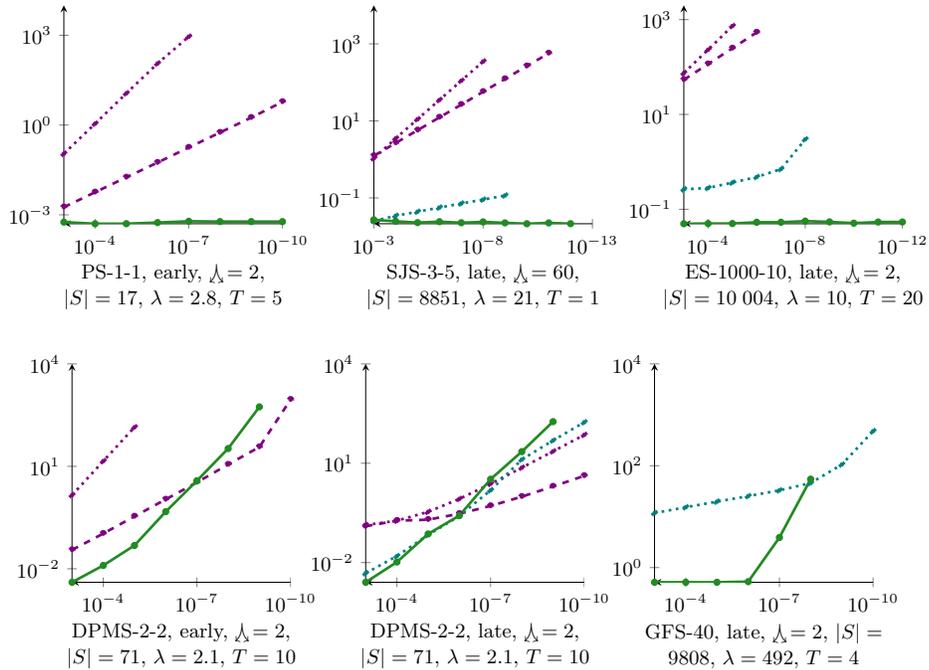

\vspace*{-3mm}
	\centering
	\begin{subfigure}[t]{\plotratio\textwidth}
		\centering
		\begin{tikzpicture}[scale=\plotscale]
		\input{figures/PollingSystem_error_early_jt1_qs1_t_5.tex}
		\end{tikzpicture}
		\label{fig:precision:a}
	\end{subfigure}%
	~ 
	\begin{subfigure}[t]{\plotratio\textwidth}
		\centering
		\begin{tikzpicture}[scale=\plotscale]
		\input{figures/SJS_Late_Precision_t=1_35.tex}
		\end{tikzpicture}
		\label{fig:precision:b}
	\end{subfigure}%
	~
	\begin{subfigure}[t]{\plotratio\textwidth}
		\centering
		\begin{tikzpicture}[scale=\plotscale]
		\input{figures/ErlangStages_Late_Precision_t=20.tex}
		\end{tikzpicture}
		\label{fig:precision:c}
	\end{subfigure}%
	\\
	\begin{subfigure}[t]{\plotratio\textwidth}
		\centering
		\begin{tikzpicture}[scale=\plotscale]
		\input{figures/DynamicPM_Early_Precision_t=10.tex}
		\end{tikzpicture}
		\label{fig:precision:d}
	\end{subfigure}%
	\begin{subfigure}[t]{\plotratio\textwidth}
		\centering
		\begin{tikzpicture}[scale=\plotscale]
		\input{figures/DynamicPM_Late_Precision_t=10.tex}
		\end{tikzpicture}
		\label{fig:precision:e}
	\end{subfigure}
	\begin{subfigure}[t]{\plotratio\textwidth}
		\centering
		\begin{tikzpicture}[scale=\plotscale]
		\input{figures/GFS_Late_Precision_t=4.tex}
		\end{tikzpicture}
		\label{fig:precision:f}
	\end{subfigure}
	\vspace*{-5mm}
	\caption{Selected experiments: Increasing precision.}
	\label{fig:precision}
\end{figure}

\begin{description}
	\item[State space.] In
	Figure~\ref{fig:space} we illustrate the effect of enlarging the state space. On the left there is a plot for early algorithms
	representing the general trend: 
\GU outperforms \ES-1 (as well as \ES-2 where applicable).
	For late algorithms in the plots on the right, the situation
	is more diverse, with \GU and \AS outperforming the \PS
	algorithms.  All algorithms exhibit similar dependency on the
	growth of the state space.

\item[Action choices.]
	Figure~\ref{fig:choice} displays the effect of increasing the
	number of actions to choose from.
For early schedulers (left) \GU generally dominates \ES-1.
For late schedulers, again \GU and \AS dominate \PS. Increasing the
	choice options in our models generally induces larger state
	spaces, so the observed growth is not to be attributed to the
	computational difficulty resulting from an increase in choice options alone.
	\item[Precision.] 	Figure~\ref{fig:precision} details precision dependency. 
Across all models, \GU works very well, excepts for some high
	precision cases, such as the DPMS models, where \ES-2 might
	be preferable over \GU in the early setting (bottom left), and
	similarly for \AS in the late setting (bottom middle). The
	same is true for the GFS case (bottom right). On the other
	hand, for some models (examples in the first row) \GU delivers very
	high precision without any runtime increase.  It is also
	interesting that generally the sensitivity of all algorithm to
	required precision is more than linear in the number of
	precision digits.
%
	\item[Time bound.]
Figure~\ref{fig:time} illustrates the effect of increasing the time bound.
	Again, the \GU-algorithm is the least sensitive in the early
	setting. For late scheduling, there are some notable QS instances
	where \PS-3 outperforms both \AS and \GU (bottom middle).
	Very large time bounds make sense only for a few models (bottom right, log-log-scale). Elsewhere, the values converge making it trivial for \ASS and \GUS.
\end{description}

\begin{figure}[t]
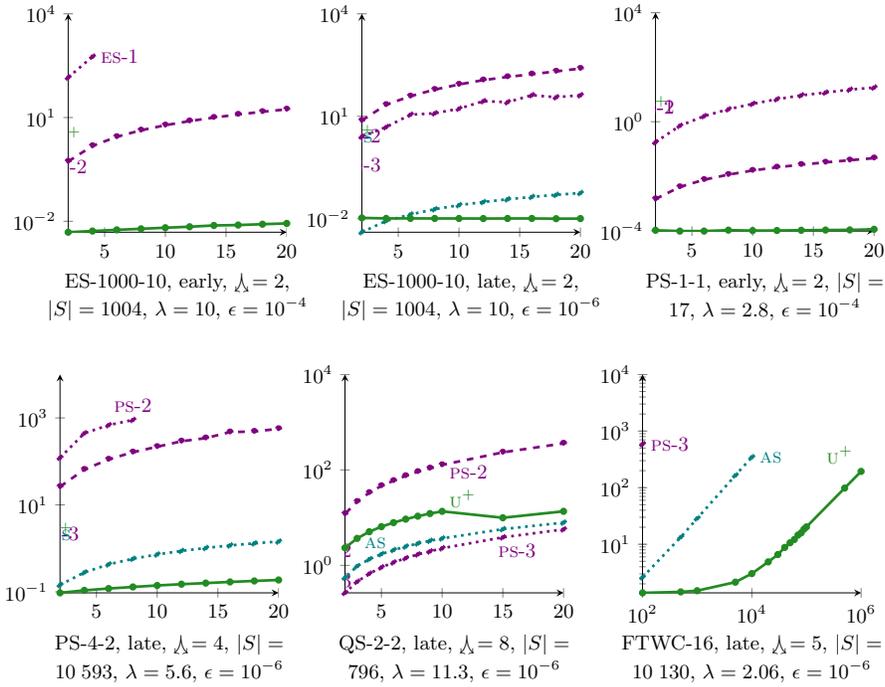

  \centering
  \begin{subfigure}[t]{\plotratio\textwidth}
    \centering
    \begin{tikzpicture}[scale=\plotscale]
      \input{figures/ErlangStages_Early_Time_e=1e-04.tex}
    \end{tikzpicture}
    \label{fig:time:a}
  \end{subfigure}%
  \begin{subfigure}[t]{\plotratio\textwidth}
    \centering
    \begin{tikzpicture}[scale=\plotscale]
      \input{figures/ErlangStages_Late_Time_e=1e-06.tex}
    \end{tikzpicture}
    \label{fig:time:b}
  \end{subfigure}%
  \begin{subfigure}[t]{\plotratio\textwidth}
    \centering
    \begin{tikzpicture}[scale=\plotscale]
      \input{figures/PollingSystem_time_early_jt1_qs1_e-04.tex}
    \end{tikzpicture}
    \label{fig:time:c}
  \end{subfigure}%
  \\
  \begin{subfigure}[t]{\plotratio\textwidth}
    \centering
    \begin{tikzpicture}[scale=\plotscale]
      \input{figures/PollingSystem_time_late_jt2_qs4_e-06.tex}
    \end{tikzpicture}
    \label{fig:time:d}
  \end{subfigure}%
  \begin{subfigure}[t]{\plotratio\textwidth}
    \centering
    \begin{tikzpicture}[scale=\plotscale]
       \input{figures/QueuingSystem_Late_Time_e=1e-06.tex}
    \end{tikzpicture}
    \label{fig:time:e}
  \end{subfigure}
  \begin{subfigure}[t]{\plotratio\textwidth}
    \centering
    \begin{tikzpicture}[scale=\plotscale]
       \input{figures/FTWC_Late_Time_e=1e-06.tex}
    \end{tikzpicture}
    \label{fig:time:f}
  \end{subfigure}
\vspace*{-5mm}
  \caption{Selected experiments: Increasing time bound.}
  \label{fig:time}
\vspace*{-7mm}
\end{figure}

\begin{wrapfigure}[7]{r}{0.25\textwidth}
  \centering
\vspace*{-7mm}
  \begin{tikzpicture}[xscale=0.7,yscale=0.7]		
	\begin{axis}[
		ymin=0,ymax=700,xmin=1, xmax=12,
		axis y line=left,axis x line=bottom,
		xlabel style={align=center},
		xlabel={},
		ylabel={},
		height=5cm,
		width=5cm,
		legend style={rounded corners,at={(.9,.9)}},
		grid style = {
			dash pattern = on 0.05mm off 1mm,
			line cap = round,
			black,
			line width = 0.5pt
		},
		cycle list name=latelist
	]
        	\addplot coordinates {
		 (1,900) 
	};

	\addplot coordinates {
		 (1,900) 
	};


	\addplot coordinates {
		 (1,10.0222) (2,21.1944) (3,33.0529) (4,45.6707) (5,58.6368) (6,71.8868) (7,85.4284) (8,99.3058) (9,113.414) (10,127.884) (25,349.183)
	} node [pos = 0.4, above]{\ASS};

	\addplot[unif] coordinates {
(1,23.9047) (2,9.97675) (3,0.411045) (4,54.3669) (5,31.1371) (5.1,13.5269) (5.2,5.44268) (5.3,2.02372) (5.4,5.88321) (5.5,79.1312) (5.6,81.0806) (5.7,6.23368) (5.8,0.743626) (5.9,0.757196) (6,0.764428) (6.1,2.43246) (6.2,17.4202) (6.3,41.9022) (6.4,94.238) (6.5,43.2158) (6.6,3.07111) (6.7,0.853101) (6.8,0.856467) (6.9,7.97852) (7,224.163) (7.1,48.0254) (7.2,231.724) (7.3,109.023) (7.4,21.7419) (7.5,8.74951) (7.6,3.56379) (7.7,9.09223) (7.8,532.279) (7.9,260.658) (8,547.055) (8.1,264.892) (8.2,125.555) (8.3,57.2064) (8.4,26.3822) (8.5,10.5032) (8.6,10.7711) (8.7,10.6565) (8.8,61.9596) (8.9,611.723) (9,139.577) (9.1,28.8813) (9.2,11.4076) (9.3,1.13958) (9.4,1.32851) (9.5,4.47242) (9.6,30.1785) (9.7,670.394) (9.8,676.28) (9.9,327.297) (10,156.133) (10.1,72.1047) (10.2,31.9896) (10.3,13.498) (10.4,5.29746) (10.5,1.3343) (10.6,1.35597) (10.7,1.51369) (10.8,1.43715) (10.9,1.44273) (11,1.72355) (11.1,1.59683) (11.2,1.58839) (11.3,1.82885) (11.4,1.54587) (11.5,1.54952) (11.6,1.55866) (11.7,1.58183) (11.8,1.59432) (11.9,1.62101) (12,1.6388) (12.1,1.79849) (12.2,1.81497) (12.3,1.7037) (12.4,1.73128) (12.5,1.74654) (12.6,1.80269) (12.7,1.77676) (12.8,2.03211) (12.9,1.82485) (13,1.8459) (13.1,1.86389) (13.2,1.88345) (13.3,1.90012) (13.4,1.91828) (13.5,1.94512) (13.6,1.95764) (13.7,1.97915) (13.8,1.99434) (13.9,2.01437) (14,2.046) (14.1,2.12216) (14.2,2.08416) (14.3,2.28318) (14.4,2.12783) (14.5,2.1395) (14.6,2.16081) (14.7,2.17822) (14.8,2.19738) (14.9,2.23218) (15,2.24743) (25,4.33769) (50,10.2865)

	} node [pos = 0.7, above left=15]{\GUS};

		\addplot (\x, {900}) node[above left]{timeout};		
	\end{axis}
  \end{tikzpicture}
\end{wrapfigure}
%
\noindent Among the many instances we considered we found a few instances where
the late \GU-algorithm shows surprising sensitivity to changes in time
bound, particularly for high precision scenarios. This is exemplified
on the right (GFS, late, $\choices = 2$, $|S| = 9808 $, $\lambda=492$,
$\epsilon=10^{-8}$, increasing time bound, no log scale). In line with the apparent general tendency of the
algorithms for
increasing parameter values, the work and thus time needed 
tends to increase
monotonously. 
Instead, small variations in time bound may lead to
great savings in runtime for \GU. This is rooted in the error
calculated while running the algorithm coincidentally falling into the
allowed margin. Less extreme examples of this behaviour are included in 
Figure~\ref{fig:precision} top row and Figure~\ref{fig:time} bottom middle.
We observed such time savings only for \GU, not for any other algorithm, 
though conceptually the runtime of \AS might profit from similar  effects as well.
The exact conditions of this behaviour are still to be found.
 
A complete list of model files, additional statistics, result tables
as well as all prototype implementations are available
at the following URL:\\ \centerline
{\url{http://depend.cs.uni-saarland.de/~hahate/atva15/}}.

\subsubsection{Evaluation and Discussion.}

The results presented show that a general
answer about the relative performance of the proposed algorithms is
not easy to give, but appears very much dependent on model parameters
outside the awareness of the modeller. Thus there is no clear winner
across all models. Still, our benchmarking, summarised in Table~\ref{tab:winner}, 
provides some general insights:
\begin{itemize}
\item All algorithms are naturally sensitive to increases in model
  parameters. Their 
runtime mostly behaves linear in the time
  bounds and the state space size, exponential in precision and superlinear (though still polynomial) in fanout.

\item For early schedulers \ES-1 is not competitive. \GU mostly outperforms \ES-2.
\item For late schedulers \PS-1 is not competitive and \PS-3 is
  effectively faster than \PS-2. \AS and \GU mostly outperform \PS-3. 
  Still each of the late algorithms $\{$\AS, \GU,
  \PS-3$\}$ is dominating the other two on at least one model
  instance. The particular algorithmic strengths have no obvious
  relation to model parameters available to the modeller. 
\item For low precision, \GU appears to be the preferred choice. For high precision, \AS is a more stable choice than \GU. 
 Yet its performance depends on non-obvious model particularities and
  algorithm parameters.
\end{itemize}

\begin{table}[t]
	\centering
{
	\begin{tabular}[t]{@{}r@{:\,}rcr@{}c@{}l@{$\!\!\!\!\!\!$}|l|r@{}l@{}}
\multicolumn{1}{c}{}		& 		\multicolumn{1}{c}{max.\ $|S|$} & 		\multicolumn{1}{c}{\parbox[c]{0.6cm}{\centering max.\\$\choices$}} & 
		\multicolumn{3}{c}{\parbox[c]{1.6cm}{\centering max.\ exit \\ rate range}}
		 & \multicolumn{1}{c}{\parbox[c]{2cm}{\centering best in early (\# of cases)}}
		  & \multicolumn{2}{c}{\parbox[c]{1.9cm}{\centering best in late (\# of cases)}}\\
		 \cline{2-9}
		{\scriptsize PS} & \numprint{743969} & \numprint{7} & \numprint{5.6}&--&\numprint{129.6}
		  & \GUS\!\textbf{(32)} & \GUS\!\textbf{(47)}\\
		{\scriptsize QS} & \numprint{16924} & \numprint{36} & \numprint{6.5}&--&\numprint{44.9}    
		& \GUS\!\textbf{(32)} &  \PSS-3\textbf{(18)}&, \GUS\!(17), \ASS(15)\\
		{\scriptsize DPMS} & \numprint{366148} & \numprint{7} &  \numprint{2.1}&--&\numprint{9.1}   
		& \GUS\!\textbf{(31)}, \ESS-2(3), \textsc{n/a}(1) & \ASS\textbf{(24)}&, \GUS\!(14), \PSS-3(6)\\
		{\scriptsize GFS} & \numprint{15258} & \numprint{2} & \numprint{252}&--&\numprint{612}
		& \GUS\!\textbf{(40)} & \ASS\textbf{(23)}&, \GUS\!(11) \\
		{\scriptsize FTWC} & \numprint{2373650} & \numprint{5} & \numprint{2}&--&\numprint{3.02}
		& \GUS\!\textbf{(25)} & \GUS\!\textbf{(32)}\\
		{\scriptsize SJS} & \numprint{18451} & \numprint{72} & \numprint{3}&--&\numprint{32}    
		& \GUS\!\textbf{(57)}, \ESS-2(2) & \GUS\!\textbf{(70)}&, \ASS(29)\\
		{\scriptsize ES} & \numprint{30004} & \numprint{2} & &\numprint{10}&   
		& \GUS\!\textbf{(23)}, \ESS-2(4), \textsc{n/a}(1)& \GUS\!\textbf{(28)}&, \PSS-3(2)\\
	\end{tabular}}

\vspace*{2mm}

	\caption{{\footnotesize Overview of experiments summarising which algorithm performed best how many times; \textsc{n/a} indicates that no algorithm completed within 15 minutes.}}
	\label{tab:winner}
\vspace*{-8mm}
\end{table}


\noindent
All in all, \GU is easy to implement for both early and late, and
competitive across a wide range of models. In settings where an
a posteriori error bound is enough, a good approximation can be usually obtained by
a variant of \GU that computes only the first iteration 
and does not increase the uniformisation rate (see the accompanying web for the error bounds obtained in experiments).

%
\section{Conclusion}

This paper has introduced \GU, a new and simple algorithm for
time-bounded reachability objectives in CTMDPs. We studied this and
all other published algorithms in an extensive comparative evaluation
for both early and late scheduling. In general, \GU performs very well
across the benchmarks, apart from late scheduling and high precision,
where it appears hard to 
predict which of the algorithms \GU, \AS,
\PS-3 performs best. One might consider to follow an approach
inspired by the distributed concurrent solver in
\textsc{Gurobi}~\cite{gurobi}. The idea is to launch all three
implementations to run concurrently on distinct cores
and report the result as soon as the first one terminates.

For researchers who want to extend an existing CTMC model checker to a
CTMDP model checker, the obvious choice is the \GU-algorithm: It
works right away for early and for late optimisation, and it requires only
a small change to the uniformisation subroutine used at the core
of CTMC model checking.
%

\begin{paragraph}{\bf Acknowledgements}
  We are grateful to Moritz Hahn (ISCAS Beijing), Dennis Guck (Universiteit Twente), and Markus
  Rabe (UC Berkeley) for discussions and technical
  contributions. 
 This work is supported
 by the EU 7th Framework Programme projects 295261
  (MEALS) and 318490 (SENSATION), by the Czech Science Foundation project P202/12/G061, the DFG Transregional Collaborative
  Research Centre SFB/TR 14 AVACS, and by the CDZ project 1023 (CAP).
\end{paragraph}

\bibliographystyle{splncs03}
\bibliography{bibliography}

\appendix
\newpage
\section{Proofs from Section~\ref{sec:our}}

We first prove the following auxiliary lemma characterizing the functions that are used to approximate the lower and upper bounds.

\begin{lemma}\label{lem:charact}
	For every $0 \leqslant k < N$ and $s\in S$, we have
	\begin{align*}
	\unvalVar[k]{s} &= \sup_{\sigma \in \TA}
	\sum_{i=k}^{N-1}  \Pr[\sinit]{\sigma}{\reach[=i]{T}{G} \mid \at{k}{s}}  \\
	\ovvalVar[k]{s} &= \sum_{i=k}^{N-1}
	\sup_{\sigma \in \TA} \Pr[\sinit]{\sigma}{\reach[=i]{T}{G} \mid \at{k}{s}}
	= \sum_{i=k}^{N-1}	
	\sup_{\sigma \in \GM_{\x}} \Pr[\sinit]{\sigma}{\reach[=i]{T}{G} \mid \at{k}{s}} \\
	\ovvalprimeVar[k]{s} &= \sup_{\sigma \in \TA} \Pr[\sinit]{\sigma}{\reach[\leqslant N-1]{}{G} \mid \at{k}{s}} 
	        = \sup_{\sigma \in \GM_{\x}} \Pr[\sinit]{\sigma}{\reach[\leqslant N-1]{}{G} \mid \at{k}{s}} 
	\end{align*}
	where $\sinit$ is the initial state, $\at{k}{s}$ are the runs that visit $s$ after $k$ steps and do not reach $G$ before $k$ steps, and $\reach[\leqslant N-1]{}{G}$ are the runs that reach $G$ within $N-1$ steps taken in arbitrary time.
\end{lemma}
\begin{proof}
	Let us first assume $s\in G$. We have $\ovvalVar[k]{s} = \unvalVar[k]{s} = \sum_{i=k}^{N-1} \pois(i)$ and $\ovvalprimeVar[k]{s} = 1$ which is also equal to the right hand sides of the equalities above. Next, we prove the equalities for $s\not\in G$ by induction. First, the first and only summand in the right hand sides above equals to $0$ while also $\ovvalVar[N-1]{s} = \unvalVar[N-1]{s} = 0$. Next, let $k < N-1$ assuming the equalities above for $k+1$. 
	\begin{align*}
	\unvalVar[k]{s} 
	= & \max_\acta \sum_{s'} \bfP^\x_\lambda(s,\acta,s') \cdot \unvalVar[k+1]{s'} \\
	= & \max_\acta \sum_{s'} \bfP^\x_\lambda(s,\acta,s') \sup_{\sigma\in\TA} \sum_{i=k+1}^{N-1} \Pr[\sinit]{\sigma}{\reach[=i]{T}{G} \mid \at{k+1}{s'}} \\
	= & \sup_{\acta,\sigma\in\TA} \sum_{i=k+1}^{N-1} \sum_{s'} \bfP^\x_\lambda(s,\acta,s') \Pr[\sinit]{\sigma}{\reach[=i]{T}{G} \mid \at{k+1}{s'}} \\
	= & \sup_{\sigma\in\TA} \sum_{i=k+1}^{N-1} \Pr[\sinit]{\sigma}{\reach[=i]{T}{G} \mid \at{k}{s}} \\
	= & \sup_{\sigma\in\TA} \sum_{i=k}^{N-1} \Pr[\sinit]{\sigma}{\reach[=i]{T}{G} \mid \at{k}{s}}
	\end{align*}
	since the first summand equals to zero.
	Similarly for $\ovvalprimeVar[k]{s}$, we have
	\begin{align*}
	\ovvalprimeVar[k]{s} 
	= & \max_{\acta} \sum_{s'} \bfP^\x_\lambda(s,\acta,s')  \ovvalprimeVar[k+1]{s'} \\
	= & \max_{\acta} \sum_{s'} \bfP^\x_\lambda(s,\acta,s')  \sup_{\sigma \in \TA} \Pr[\sinit]{\sigma}{\reach[\leqslant N-1]{}{G} \mid \at{k+1}{s'}} \\
	= & \sup_{\sigma \in \TA} \Pr[\sinit]{\sigma}{\reach[\leqslant N-1]{}{G} \mid \at{k}{s}}
	\end{align*}
	and the same hold when ranging over schedulers in $\GM_{\x}$.
	Finally, for $\ovvalVar[k]{s}$, 
	\begin{align*}
	\ovvalVar[k]{s}
	=& 
	\sum_{i=k}^{N-1} \pois(i) \cdot \ovvalprimeVar[(N-1)-(i-k)]{s} \\
	=& 
	\sum_{i=k}^{N-1} \pois(i) \cdot \sup_{\sigma \in \TA} \Pr[\sinit]{\sigma}{\reach[\leqslant N-1]{}{G} \mid \at{(N-1) - i + k}{s}} \\
	=& 
	\sum_{i=k}^{N-1} \pois(i) \cdot \sup_{\sigma \in \TA} \Pr[\sinit]{\sigma}{\reach[\leqslant i]{}{G} \mid \at{k}{s}} \\		
	=& 
	\sum_{i=k}^{N-1} \sup_{\sigma \in \TA} \Pr[\sinit]{\sigma}{\reach[=i]{T}{G} \mid \at{k}{s}}
	\end{align*}
	and again the same hold when ranging over schedulers in $\GM_{\x}$.\qed
\end{proof}

\begin{reflemma}{lem:lower-upper}
	It holds that $\vale{}\leq\vall{}$, and for any CTMDP $\C^\x_\lambda$, $\unval{}\leq\val{}\leq\ovval{}$.
\end{reflemma}
\begin{proof}
	$\vale{}\leq\vall{}$ follows directly from the fact that $\GM_\early \subseteq \GM_\late$.
	$\unval{}\leq\ovval{}$ since 
	
	$$\sup_{\sigma \in \GM_\x}
	\sum_{i=0}^\infty  \Pr[s]{\sigma}{\reach[=i]{T}{G}}
	\leqslant 
	\sum_{i=0}^\infty 
	\sup_{\sigma \in \GM_{\x}} \Pr[s]{\sigma}{\reach[=i]{T}{G}}$$ (optimizing for each subset separately yields higher value). \\
Furthermore $\TA \subseteq \GM_\early$ implies $\unval{}\leq\vale{}$, and for each $i \in \Nseto$, it follows from Lemma~\ref{lem:charact} that
	$$\sup_{\sigma \in \GM_{\x}} \Pr[s]{\sigma}{\reach[=i]{T}{G}} = \sup_{\sigma \in \TA} \Pr[s]{\sigma}{\reach[=i]{T}{G}}.$$ 
\end{proof}

\begin{reflemma}{lem:approximations}
In any CTMDP $\C^\x_\lambda$, $\lVert \unvalVar[0]{} - \unval{} \rVert_\infty \leq \eps \cdot \kappa$ and $\lVert \ovvalVar[0]{} - \ovval{} \rVert_\infty \leq \eps \cdot \kappa$.
\end{reflemma}

\begin{proof}
	The proof follows from Lemma~\ref{lem:charact}. It is completed by the observation~\cite{DBLP:journals/iandc/BrazdilFKKK13} that the probability of $\geqslant N$ steps to be taken within $T$ is $\leqslant \eps \cdot \kappa$. This is independent of the scheduler and hence, for both $\unvalVar[0]{}$ and $\ovvalVar[0]{}$ we obtain the desired error bound. \qed
\end{proof}

\begin{reflemma}{lem:converge}
	We have $\displaystyle \lim_{\lambda \to \infty} g_\lambda \to 0$ where 
	$g_\lambda$ denotes the gap $\lVert \unval{} - \ovval{} \rVert_\infty$ in $\C^\x_\lambda$.	
\end{reflemma}

\begin{proof}
	From Lemmata~\ref{lem:lower-upper} and~\ref{lemma:unif-preserves} we have that for any $\lambda$, $\unval{} \leqslant \val{} \leqslant \ovval{}$. It remains to show that for any state $s \in S$, we have
	\begin{align}\label{eq:goal}
	\lim_{\lambda \to \infty} \lvert \unval{s} - \ovval{s} \rvert \to 0.
	\end{align}
	Let $\lambda_0$ be the maximal exit rate in $\C$ and let $\eps > 0$. We need to find a uniformisation rate $\lambda = k\lambda_0$ such that $\lvert \unval{s} - \ovval{s} \rvert  \leqslant \eps$.
	Consider Chebyshev inequality
	 \begin{align*}
		 Pr[|\psi_{\lambda T} - \lambda T| \geqslant m \sigma] \leqslant \frac{1}{m^2}
	 \end{align*}
	where $\sigma$ is standard deviation of $\psi_{\lambda T}$ and $m > 0$. Let $\frac{1}{m^2} = \frac{\varepsilon}{6}$. Then Chebyshev inequality for $\psi_{\lambda T}$ can be written as
	 \begin{align*}
		 Pr[|\psi_{\lambda T} - \lambda T| \geqslant \sqrt{\frac{6 \lambda T}{\varepsilon}}] \leqslant \frac{\varepsilon}{6}
	 \end{align*}
	or
	 \begin{align*}
		 Pr[\psi_{\lambda T} \in (\lambda T - \sqrt{\frac{6 \lambda T}{\varepsilon}}, \lambda T + \sqrt{\frac{6 \lambda T}{\varepsilon}}) ] > 1-\frac{\varepsilon}{6}
	 \end{align*}	
	
	
	For any uniformisation rate $\lambda = k \cdot \lambda_0$, we define $a_\lambda = \lfloor \lambda T -\sqrt{6\lambda T / \eps} \rfloor$ and $b_\lambda = \lceil \lambda T + \sqrt{6\lambda T / \eps} +1\rceil$.
	Then,
	\begin{align}\label{eq:bounds}
	\sum_{i=a_\lambda}^{b_\lambda-1} \pois(i) > 1- \frac{\eps}{6}.
	\end{align}
	In the uniformised model $\C^\x_{k \cdot \lambda_0}$, the probability of not changing state in one step is $$\frac{\lambda-E(s, \alpha)}{\lambda} \geqslant \frac{\lambda-\lambda_0}{\lambda} = \frac{k-1}{k}$$
	Therefore, for $c_\lambda = b_\lambda - a_\lambda \leqslant 2 (\sqrt{6\lambda T / \eps}+1)$, the probability $p_\lambda$ of not changing state at all within $c_\lambda$ steps satisfies 
\begin{align*}
	p_\lambda	&= \left(\frac{k-1}{k}\right)^{c_\lambda} \geqslant \left(\frac{k-1}{k}\right)^{2 (\sqrt{6k\lambda_0 T / \eps}+1)} = \\
				&= \left[\left(\frac{k-1}{k}\right)^{\sqrt{k}}\right]^{2 (\sqrt{6\lambda_0 T / \eps})}\left(\frac{k-1}{k}\right)^{2} = \\
				&= \left[\left(1-\frac{1}{\sqrt{k}}\right)^{\sqrt{k}} \left(1+\frac{1}{\sqrt{k}}\right)^{\sqrt{k}}\right]^{2 (\sqrt{6\lambda_0 T / \eps})}\left(\frac{k-1}{k}\right)^{2}
\end{align*}
Thus
	\begin{align}\label{eq:converge}
		\lim \limits _{k \to \infty} p_\lambda	&= \left[\frac{1}{e} \cdot e \right]^{2 (\sqrt{6\lambda_0 T / \eps})} = 1
	\end{align}
	Instead of (\ref{eq:goal}) we prove that there is $\lambda$ such that
		$$
		\lvert  \ovvalVarb^{b_\lambda}_0(s) - \unvalVarb^{b_\lambda}_0(s) \rvert \leqslant \eps/2.
		$$
	where $\unvalVarb^{b_\lambda}_i(s)$, $\ovvalprimeVarb^{b_\lambda}_i(s)$, and $\ovvalVarb^{b_\lambda}_i(s)$ is defined for all $k$ as $\unvalVarb$, $\ovvalprimeVarb$, and $\ovvalVarb$, only replacing $N$ by $b_\lambda$. This suffices, as from proof of Lemma~\ref{lem:approximations}, the approximations with upper bound $b_\lambda$ are only more precise than approximations with upper bound $N$. Using (\ref{eq:converge}), we fix $\lambda$ to be such that $p_{\lambda} \geqslant 1 - \eps/6$.
	Then, we have
	\begin{align*}
	\unvalVarb^{b_\lambda}_0(s) 
	& \geqslant \unvalVarb^{b_\lambda}_{a_{\lambda}}(s) 
	\intertext{and from (\ref{eq:bounds}), we can obtain by straightforward induction}
	& \geqslant \underline{u}^{b_\lambda}_{a_{\lambda}}(s) - \frac{\eps}{6}
	\intertext{where $\underline{u}_k$ is defined as $\unvalVar[k]{}$ except for goal states having value $1$ instead of the sum of poisson probabilities. The term $\underline{u}^{b_\lambda}_{a_{\lambda}}(s)$ now equals by definition the term $\ovvalprimeVarb^{b_{\lambda}}_{(b_{\lambda}-1)-a_\lambda}(s)$, and hence}
	& = \ovvalprimeVarb^{b_{\lambda}}_{(b_{\lambda}-1)-a_\lambda}(s) - \frac{\eps}{6}
	\; \geqslant \; \sum_{i=a_\lambda}^{b_\lambda-1} \pois(i) \cdot \ovvalprimeVarb^{b_{\lambda}}_{(b_{\lambda}-1)-a_\lambda}(s) - \frac{\eps}{6};
	\intertext{Furthermore, from the fact that $p_\lambda \geqslant 1- \eps/6$, we obtain that each $\ovvalprimeVarb^{b_{\lambda}}_{(b_{\lambda}-1)-i}(s) \leqslant \ovvalprimeVarb^{b_{\lambda}}_{(b_{\lambda}-1)-a_\lambda}(s) - \eps/6$ and}
	& \geqslant  \sum_{i=a_\lambda}^{b_\lambda-1} \pois(i) \cdot \ovvalprimeVarb^{b_{\lambda}}_{(b_{\lambda}-1)-i}(s) - \frac{2\eps}{6} 
	\geqslant \sum_{i=0}^{b_\lambda-1} \pois(i) \cdot \ovvalprimeVarb^{b_{\lambda}}_{(b_{\lambda}-1)-i}(s) - \frac{3\eps}{6} 
	\intertext{and finally, we get by definition}
	& = \ovvalVarb^{b_\lambda}_0(s) - \eps/2.
	\end{align*}
	\qed
\end{proof}

\noindent Let $\mathcal{D}(A)$ denote the set of all distributions over a discrete set $A$. Then

\begin{definition} A stochastic update scheduler $\sigma$ on a CTMDP $\C=(S,\Act,\bfR)$ is a tuple $\sigma = (\mathcal{M}, \sigma_u, \pi_0)$, where
\begin{itemize}
	\item[--] $\mathcal{M}$ is a countable set of memory elements
	\item[--] $\sigma_u:\mathcal{M} \times S \times \mathbb{R}_{\geqslant 0} \mapsto \mathcal{D}(\mathcal{M}, \mathit{Act})$ is the update function
	\item[--] $\pi_0:S \mapsto \mathcal{D}(\mathcal{M})$ distribution over initial memory values
\end{itemize}
\end{definition}

The system operates under a stochastic update scheduler as follows. At first initial memory values are sampled from the distribution $\pi_0(s)$. Afterwards, given current memory value, current state and time spent in the state so far (not the time from the beginning of the process), the stochastic update function $\sigma_u$ continuously updates the memory value and the action to be taken. When the system decides to leave the state, upon entering the successor state the memory is also updated by one.

\begin{reflemma}{lemma:stoch-update}
The values $(\unvalVar[k]{})_{0\leq k \leq N}$ computed by Algorithm~\ref{alg:unif} for given $\C$, $\x$, and $\eps > 0$ yield a stochastic update scheduler $\widetilde{\sigma}^{\x}_{TD}$ that is $\varepsilon$-optimal in $\C$.
\end{reflemma}
\begin{proof}

\noindent Let $\C=(S, \Act,\bfR)$ be the original CTMDP, $\lambda$ - the uniformization rate computed by Algorithm~\ref{alg:unif}, $\C^\x_\lambda=(S_{\lambda}, \Act,\bfR_{\lambda})$ - CTMDP uniformised with rate $\lambda$. Computation of the lower bound $\unvalVar[0]{}$ involves as well computation of the $\varepsilon$-optimal scheduler that attains the bound.  Let $\widetilde{\sigma}_{\TA}^{\x}: S_{\lambda} \times \mathbb{N_0}_{\geqslant 0} \mapsto \mathit{Act}$ be this scheduler and $i\C_{\TA}^{\x}=(iS_{\TA}, i\bfR_{\TA}^{\x})$ - the CTMC induced by $\widetilde{\sigma}_{\TA}^{\x}$, where $iS_{\TA} = S_{\lambda} \times \mathbb{N}_0$. Then,

$$i\bfR_{\TA}(s_1, s_2) = 
					\left\{
						\begin{array}{cll}
							\bfR(s', \widetilde{\sigma}_{\TA}^{\x}(s', m), s'')		&\mbox{ if } s_1 = (s', m) \text{ and } s_2=(s'', m+1) \\ \\
							\lambda-E(s', \widetilde{\sigma}_{\TA}^{\x}(s', m)) & \mbox{ if } s_1 = (s', m) \text{ and } s_2=(s', m+1) \\ \\
							0 & \mbox{ otherwise }\\
						\end{array}
					\right.\\
$$

Let $\pi_{(s, m)}(s', m+k, t)$ be the transient probability in $i\C_{\TA}^{\x}$ for state $(s, m+k)$, given that the system starts from state $(s, m)$. Then

$$\pi_{(s, m)}(s', m+k, t) = 
					\left\{
						\begin{array}{lll}
							e^{-\lambda t} \frac{(\lambda-E_0) \cdots (\lambda-E_{k-2}) (\lambda-E_{k-1}) }{k!}t^k		&\mbox{ if } s' = s\\
							\\
							e^{-\lambda t} \frac{(\lambda-E_0) \cdots (\lambda-E_{k-2}) \bfR(s, \alpha_{k-1}, s') }{k!}t^k & \mbox{ otherwise } \\
						\end{array}
					\right.\\
$$
where $\alpha_i=\widetilde{\sigma}_{\TA}^{\x}(s, m+i))$ and $E_i=E(s, \alpha_i)$. 

W.l.o.g. we assume that after $N$ transitions have been performed the scheduler $\widetilde{\sigma}_{\TA}^{\x}$ takes the same decision for every state, irrespectively of the memory value. We denote this decision as $\alpha_{N}(s)$. We now define the finite memory stochastic update scheduler $\widetilde{\sigma}_{\GM} = (\mathcal{M}, \sigma_u, \pi_0)$:
\begin{itemize}
	\item[--] $\mathcal{M} = [0..N] \cup \bot$
	\item[--] 
			$			
			\forall m, m' \in \mathcal{M}, m, \neq \bot, \ s \in S, \ t \in \mathbb{R}_{\geqslant 0} \\ \\
				\sigma_u^{\x}(\bot, s, t) := \left[(\bot, \alpha_N(s)) \mapsto 1, \mathit{otherwise} \mapsto 0\right] \\ \\
				\sigma_u^{\x}(m, s, t)(m', a) := 
					\left\{
						\begin{array}{lll}
							\pi_{(s, m)}(s, m', t)
								& \mbox{if } m' \in [m+1 .. N] \text{ and } \\
									&a=\widetilde{\sigma}_{\TA}^{\x}(s, m^{\x}), \text{where }\\
									&m^{\early}=m, m^{\late}=m' \\ \\
							\sum \limits_{k = 1}^{\infty} \pi_{(s, m)}(s, N+k, t)
								& \mbox{if } m' = \bot \text{ and } \\
									&\x = \early \text{ and } a = \widetilde{\sigma}_{\TA}^{\early}(s, m) \text{ or } \\
									&\x = \late \text{ and } a = \alpha_N(s)	\\ \\
							0		& \mbox{otherwise} \\
						\end{array}
					\right.\\ \\
				$
	\item[--] $\forall s \in S \ \pi_0(s):=\left[ 0 \mapsto 1, \text{otherwise} \mapsto 0 \right]$
\end{itemize}

Intuitively, when the system moves to a state $s$ and memory $m$ has been collected up until this moment, it updates the memory according to the sub-process $\C^\x_\lambda(s, m)$ of $\C^\x_\lambda$ while residing in $s$ and the memory update is finished when $\C$ decides to leave $s$. This sub-process $\C^\x_\lambda(s, m)$ is a process that starts when $\C^\x_\lambda$ moves to the state $s$ and evolves when the uniformized system takes introduced high-rate transitions. Thus, the value of the memory is equivalent to the length of the history of the uniformized process, i.e. the $\widetilde{\sigma}_{\GM}$ simulates evolution of the uniformized process and takes exactly the decisions that $\widetilde{\sigma}_{\TA}$ would take. Thus, the transient distribution of the processes induced by $\widetilde{\sigma}_{\GM}$ and $\widetilde{\sigma}_{\TA}$ are exactly the same.

The amount of memory used by $\widetilde{\sigma}_{\GM}$ is in $O(N |S|)$.

\end{proof}

\end{document}